\documentclass[]{article}
\pdfoutput=1
\usepackage{graphicx}
\usepackage[utf8]{inputenc}
\usepackage{amssymb}
\usepackage{amsmath}
\usepackage{amsthm}
\usepackage{tikz}
\usepackage{amsmath}
\usepackage{tabularx}
\usepackage[ngerman]{babel}
\usepackage{float}
\usepackage[margin=10pt, tableposition=top]{caption}

\usetikzlibrary{matrix}
\usetikzlibrary{cd}
\usetikzlibrary{babel}

\theoremstyle{definition}
\newtheorem{definition}{Definition}[subsection]
\newtheorem{example}{Beispiel}[subsection]
\newtheorem{bemerkung}{Bemerkung}[subsection]
\theoremstyle{plain}

\newtheorem{lemma}{Lemma}[subsection]

\title{Phänomen-Signal-Modell: Formalismus, Graph und Anwendung}
\author{Hans Nikolaus Beck\footnote{Robert Bosch GmbH, Abstatt}  \and Nayel Fabian Salem \footnote{Technische Universität Braunschweig, Institut für Regelungstechnik} \and Veronica Haber\footnote{PROSTEP AG, München} \and Dr.-Ing. Matthias Rauschenbach\footnote{Fraunhofer-Institut für Betriebsfestigkeit und Systemzuverlässigkeit LBF, Darmstadt} \and Jan Reich\footnote{Fraunhofer-Institut für Experimentelles Software Engineering IESE, Kaiserslautern}}
\begin{document}

\maketitle

\begin{abstract}
Betrachtet man Information als Grundlage des Handelns, so wird es interessant sein, Fluss und Erfassung von Information zwischen den Akteuren des Verkehrsgeschehens zu untersuchen. Die zentrale Frage ist, welche Signale ein Automat im Straßenverkehr empfangen, decodieren oder senden muss, um konform zu geltenden Maßstäben und sicher zu agieren. Das Phänomen-Signal-Modell ist eine Methode, das Problemfeld zu strukturieren, eben diesen Signalfluss zu analysieren und zu beschreiben. Der vorliegende Aufsatz erklärt Grundlagen, Aufbau und Anwendung dieser Methode.
\end{abstract}

\section{Motivation und Aufgabe}\label{sec:motivation}

Im Projekt \emph{\glqq Verifikation und Validierung von autonomen Fahrzeugen L4/L5\grqq{}} (kurz
	\emph{VVM}\footnote{Das Projekt Verifikation und Validierung autonomer Fahrzeuge L4/L5 ist Teil der VDA Leitinitiative und Teil der PEGASUS Projekt Familie, vom Bundesministerium für Wirtschaft und Energie (BMWi, http://www.bmwi.de) gefördert. Web: www.vvm-projekt.de})
 ist die grundlegende Fragestellung, wie  Sicherheit und gesellschaftliche bzw. gesetzliche Konformität eines hochautomatisierten oder autonomen Fahrzeugs methodisch erreicht und begründet werden kann. VVM setzt auf den Ergebnissen des PEGASUS Projektes auf. 
 
 Aus dem Anspruch, der Automat solle ein Mitglied der sozialen Verkehrsgesellschaft sein, sich also ohne Zutun eines menschlichen Fahrers im urbanen Umfeld (konkret: Kreuzungen) bewegen, wird die schiere Zahl möglicher Geschehnisse und Situationen zum Hindernis bisheriger Ingenieursverfahren. Einen Beitrag, wie mit diesem als \emph{Open Context} bezeichneten Aspekt umgegangen werden könnte, ist Motivation und Ziel dieses Aufsatzes. 

In Bezug auf die Konstruktion eines solchen Automaten steht zu Beginn, wie in jedem Entwicklungsprozess, die Erhebung von Anforderungen. Gemäß der Klassifikation nach SAE \cite{sae}  entspricht dem Bau eines Automaten mit der Befähigung, Teil eines vom Menschen dominierten interaktiven Systems zu sein, die Konstruktion eines Level 4/5 Systems. Anforderungen formuliert als eine Sammlung von Systemfunktionen genügen augenscheinlich nicht. Der Grund liegt in der unendlichen Varianz möglicher Situationen. Der für die Automobilindustrie zentrale Sicherheitsstandard ISO 26262 \cite{iso3} \cite{iso4} \cite{iso5} ist eine Manifestation der Vorstellung, dass Systemfunktionen getrennt gedacht und die Bedingungen und Ansprüche ihres Einsatzes systematisch bestimmt werden können. Im SOTIF Standard \cite{sotif} wird  die Erkenntnis der Industrie evident, dass im Open Context eben diese Bedinungen für den Einsatz jener Systemfunktionen mannigfaltig sein kann. Das System wird also komplexe Verhaltensmuster generieren müssen, in denen je nach Situation unterschiedliche Kombinationen von Funktionen aufzurufen sind. Damit wird die Bestimmung von Anforderungen im Open Context eine Suche nach Verhalten, genauer Sollverhalten, welches zu finden und zusammen mit den notwendigen Fähigkeiten und Eigenschaften des Systems zu beschreiben diese Arbeit zum Ziele hat. 

Verkehrsgeschehen ist ganz wesentlich auch Kommunikationsgeschehen. Da der Automat selbst kein Subjekt ist, aber mit Subjekten im Verkehr interagieren soll, ist Anlass zu  dem Ansatz gegeben, Informations-Signalflüsse im Verkehr zu untersuchen. Durch eine entsprechende Modellbildung soll dabei erreicht werden, dass sich die hier dargestellte Formalisierung algorithmisieren und in ein Computerprogramm umsetzen lässt. Als Grundbaustein dafür wird es auch notwendig sein, das Verhalten der menschlichen Verkehrsteilnehmer, also das Verhalten nach Gesetz oder gesellschaftlichen Maßstäben, künftig als \textit{Normverhalten} bezeichnet, zu identifizieren.

Der Philosoph Edmund Husserl hat sich im Rahmen seiner Philosophie \cite{EdHrl} detaillierte Gedanken zu der Frage gemacht, wie Wahrnehmung und Kommunikation zwischen Subjekten geschehen. Ein zentrales Ergebnis ist, dass die Intentionalität von Zeichen und Worten, aber auch die bisherige Erfahrung der Subjekte konstituierend ist für die Art und Weise, wie jene Zeichen und Worte verstanden werden. Husserls Arbeiten legen zwingend nahe, dass der Informationsfluss allein nicht Handeln erklären kann. Gehörte oder gesehene Information erhält Bedeutung auf der Grundlage der individuellen Geschichte des erfassenden Subjektes, Information wird somit intentional (Sender) und subjektiv (Empfänger). 

Im Rahmen der Forschung zur Künstlichen Intelligenz und den daraus hervorgegangenen Agentenmodellen sind verschiedene Anteile des handelnden, entscheidenden Menschen in den Blick gerückt und Modellierungen entworfen worden. Grundfrage vieler Arbeiten ist, wie eine Maschine ihr Agieren nach einem Ziel ausrichtet. Dies führte zu diversen, an mathematischer Logik orientierten Algorithmen (siehe z.B. \cite{bradko} \cite{schalkoff} \cite{agenten}).

Mittels der Erweiterung der logischen Mittel, wie etwa dem Situation Calculus  \cite{Raymond}, werden Ansätze verfolgt, zusätzlich zu Zielen auch Situation und Ereignisse logisch in die Aktionsentscheidung einzubeziehen. Eine andere Anpassung betrifft die modale Logik, die die Kategorien \glqq notwendig\grqq{} und   \glqq möglich\grqq{} hinzunimmt. Damit wird und wurde versucht, das Wissen im Agenten und dessen Perzeption zu modellieren \cite{Hintikka}. Ergänzt werden all diese Bestrebungen durch den Versuch, psychologische Erkenntnise über menschliches Wissen und Entscheidungsfindung mathematisch zu modellieren (stellvertretend \cite{Tenenbaum} \cite{CompPsych}). 

\section{Ansatz}\label{sec:ansatz}

Verkehrsgeschehen als Komunikationsgeschehen und  Verkehrsgeschehen als durch Rechtsvorschriften geregelter Raum, die Suche nach und die Beschreibung von Sollverhalten, diese Gegebenheiten fordern ein Modell, welches sowohl die subjektive Seite des einzelnen Verkehrsteilnehmers mit seinem Vorwissen, als auch die Objektivierung hin zu Gesetzestexten und den Bau eines Automaten beschreiben kann.  In der Abwägung dieser Gründe erscheinen dann Husserls Ideen als attraktivster Ansatz für ein Modell, um die oben skizzierten Aufgaben anzugehen und einen wesentlichen Beitrag zu deren Lösung zu leisten.   

Vor Darlegung des Ansatzes sei zunächst eine allgemeine Bemerkung angebracht: Alle hier beschriebenen Betrachtungen beziehen sich auf ein Szenario der im Rahmen des VVM Projektes gewählten Anwendungsfälle, zu sehen in Abbildung \ref{fig:szenariosmall}. Zonen wie die gezeigten b1, b2 etc. sind ein geeignetes Modellierungsmittel,  Räume (oder besser: Flächen) eines Szenarios nach Relevanz einzuteilen. Die Art und Weise zu beschreiben, wie derartige Zonen systematisch abgeleitet werden können, würde jedoch den Rahmen dieses Beitrags sprengen, vgl. hierzu \cite{Butz}.

\begin{figure} [H]
	\centering
	\includegraphics[width=0.5\linewidth]{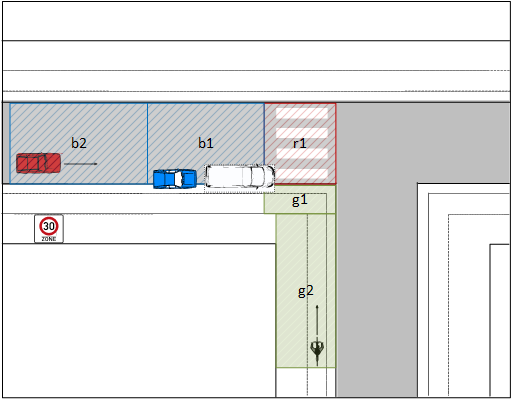}
	\caption[Abbildung]{Beispiel-Szenario, Ego ist rot }
	\label{fig:szenariosmall}
\end{figure}

Husserls Phänomenologie  \cite{EdHrl} als Grundlage dieser Arbeit legt  konsequenter Weise den Namen \emph{Phänomen-Signal-Modell}, oder kurz \emph{PSM}, für den hier dargstellten Ansatz nahe. Vereinfacht gesagt, gliedert sich das PSM in folgende Teile:

\begin{enumerate}
	\item Information, Wissen und Prognosen sind konstituierend für Aktionen. Der Weg, wie aus Erfassungen durch Sinne oder Sensoren diese Elemente gebildet werden, ist mittels der Ideen der Phänomenologie  zu beschreiben.
	\item Welche Aktionen auf der Basis jener Elemente möglich oder notwendig erscheinen, ist transparent zu beschreiben.
\end{enumerate}

Im Kontext des PSM werden diese Aufgaben erschlossen und einer Lösung zugänglich gemacht, indem Aktionen (oder Handlungen bei Subjekten) als \textit{Regeln} formuliert werden. Dies ist eine konstruktive Entscheidung. Zuvor sollen noch einige für den Ansatz relevante Begriffe verabredet sein.

\section{Begrifflichkeiten}\label{sec:begriffe}

In der Begriffsbildung wurde streng darauf geachtet, ob man eine Aussage über den Automaten oder einen Menschen treffen will. Ein Mensch kann wohl \textit{wahrnehmen}, im Falle des Automaten sprechen wir ausschliesslich von \textit{Messen} oder \textit{Erfassen}. Ebenso \textit{agiert} ein Automat, von \textit{Handeln} wird im Rahmen dieser Arbeit nur im Kontext von Menschen gesprochen. Einige Aspekte des maschinellen Handelns werden in \cite{Missel} beleuchtet, wovon das Konzept der Selbstursprünglichkeit besonders interessant erscheint, um vorprogrammiertes Verhalten von dem eines autonomen Systems zu unterscheiden. Selbstursprünglichkeit ist die Idee, dass ein System weitere Gründe des Interagierens mit der Umwelt haben kann, als reine Reiz-Reaktions- oder Input-Output Mechanismen. Das PSM ist jedoch nicht als Hilfsmittel gedacht, diese Probleme zu untersuchen, der Fokus liegt auf der Erfassungs-Signal-Wissenskette. Damit liegt das PSM der heute darstellbaren Technik auch näher. 

Der Begriff \textit{Phänomen} ist natürlich zentral für Husserls Arbeit. Hier kann nicht die detaillierte Definition von Husserl sinnvoll sein, da diese umfangreich und nicht leicht zu verstehen ist. Darum sei \textit{Phänomen} im Einklang mit \cite{Brd} S436ff als eine beobachtbare und in Bezug auf die eigene Intention relevante Gegebenheit der Umwelt verstanden. Wenn \textit{Information} als eine Menge von Daten mit Beziehung untereinander aufgefasst werden kann, dann ist ein \textit{Signal} eine Information von Bedeutung für das wahrnehmende Subjekt oder die erfassende Maschine. Zusammen mit dem vorhandenen Wissen aus dem bisherigen Erleben (Subjekt) oder eingegebenen Wissen (der Maschine) leistet ein Signal einen Beitrag für die Bildung von Prognosen und Aktionsentscheidungen. 

Dazu ein Beispiel: sehe ich als Autofahrer (der die StVO\footnote{Staßenverkehrsordnung} kennt und Erfahrung hat) die Ampel auf Gelb springen, so ist dieses ein Signal, woraus ich die Prognose stellen kann, dass ich meine Vorfahrtsberechtigung verlieren werde. 

Die Abbildung \ref{fig:wissenundoperator} visualisiert im Ergebnis des Beschriebenen den Kreislauf von Erfassung, Signal, Wissen und Prognose bzw. Aktion, welcher in Konsequenz das anzeigt, was ein Phänomen ausmacht. Durch das Zusammenspiel dieser Faktoren wird ein Erfassen von etwas aus der Umwelt zu einem bewussten (Subjekt) oder verarbeitbaren (Maschine) Phänomen. Das dort mehrfach gezeigte Symbol $\alpha A$ steht für eine Wirkung $\alpha$ der Aktion A. All diese Elemente werden im Abschnitt \ref{sec:sachverhaltebasis} weiter erläutert.

\begin{figure}
	\centering
	\includegraphics[width=0.85\linewidth]{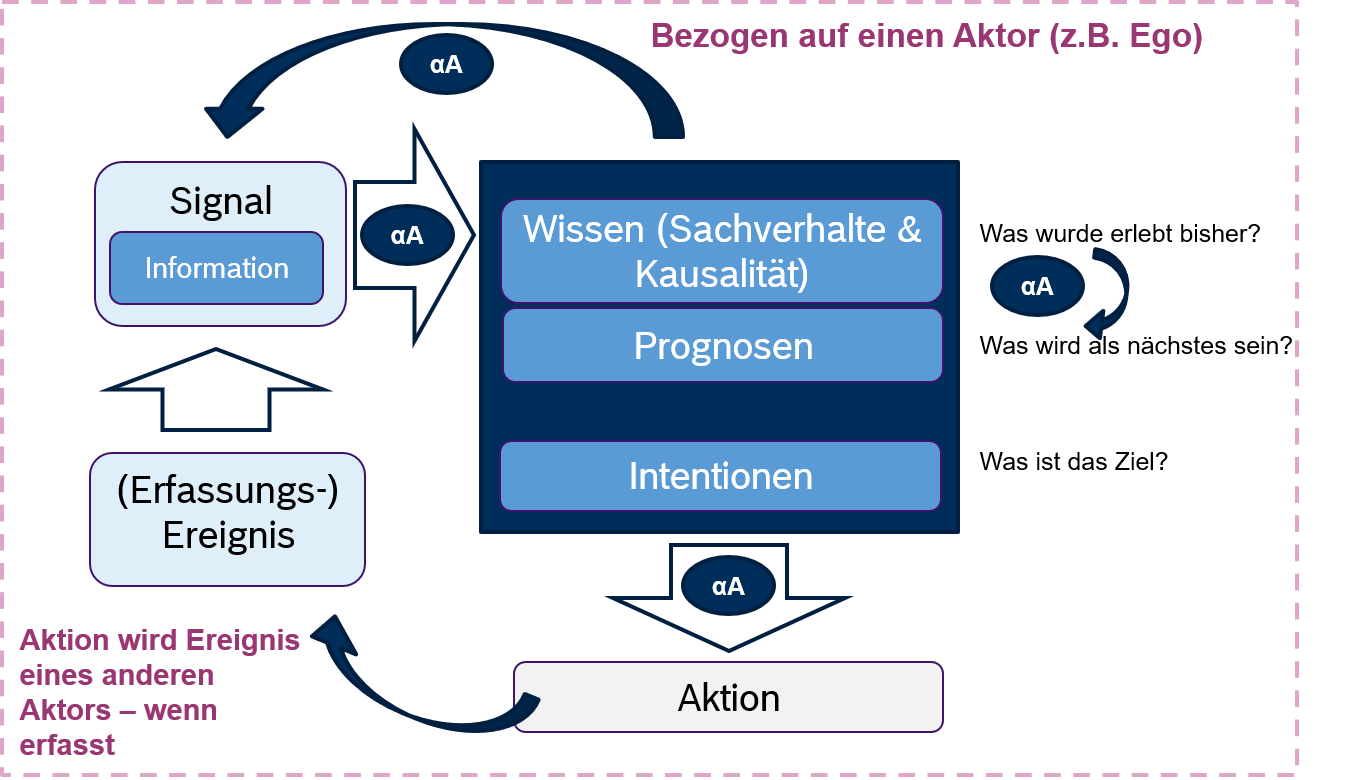}
	\caption[Ein Bild]{Erfassung-Signal-Wissen Kreislauf}
	\label{fig:wissenundoperator}
\end{figure}

Regeln, das wesentliche Mittel in dieser Arbeit, erfordern klar formulierte Bedingungen ihres Eintretens. In ähnlicher Weise, wie die Rechtswissenschaften, welche ja wesentlich das Normverhalten mitbestimmen,  Rechtsfolgen und Tatbestände in Gesetzen miteinander verknüpfen  \cite{Eng}, sind \textit{Sachverhalte} all diejenigen festgestellten Gegebenheiten, welche die Anwendung einer Regel begründen, also die \glqq WENN\grqq{} Bedingung darstellen. Und ebenfalls wie in den Rechtswissenschaften führen Indizien oder Tatbestandsmerkmale, hier \textit{Indikatoren} genannt, zu der Feststellung eines Sachverhaltes. 

\glqq Wissen\grqq{} im Rahmen des PSM hat somit zwei Aspekte: einerseits bezeichnet es das Vorwissen im Automaten, andererseits entspricht es dem erfassten und erkannten Sachverhalt. Natürlich ist es denkbar, dass dieser erkannte Sachverhalt zu neuem Wissen im Automaten führt, was letztlich nichts anderes als Lernen bedeutet. Dieser Aspekt wird in Folgearbeiten beleuchtet werden. Weil die Regeln, die dem Automaten gegeben werden, auf Sachverhalten gründen, kann gesagt werden, dass im Rahmen des PSM Wissen durch Sachverhalte dargestellt werden können.

Abbildung \ref{fig:sachverhalte} illustriert nochmals diese für diesen Ansatz wesentlicher Zusammenhänge. Notwendige \emph{Fähigkeiten} sind entsprechend jene, die zur Feststellung von Sachverhalten und den mittels Regeln implementierten gewünschten Aktionen, also dem \emph{Sollverhalten} notwendig sind. Zugunsten einer exakten Betrachtung des Begriffs \emph{Fähigkeit} im Kontext von Maschinen wird auf Arbeiten anderer Arbeitsgruppen in VVM verwiesen \cite{Jatz} \cite{Nolte} \cite{Reschka} \cite{Bagschik}. 

\begin{figure}[h]
	\centering
	\includegraphics[width=0.8\linewidth]{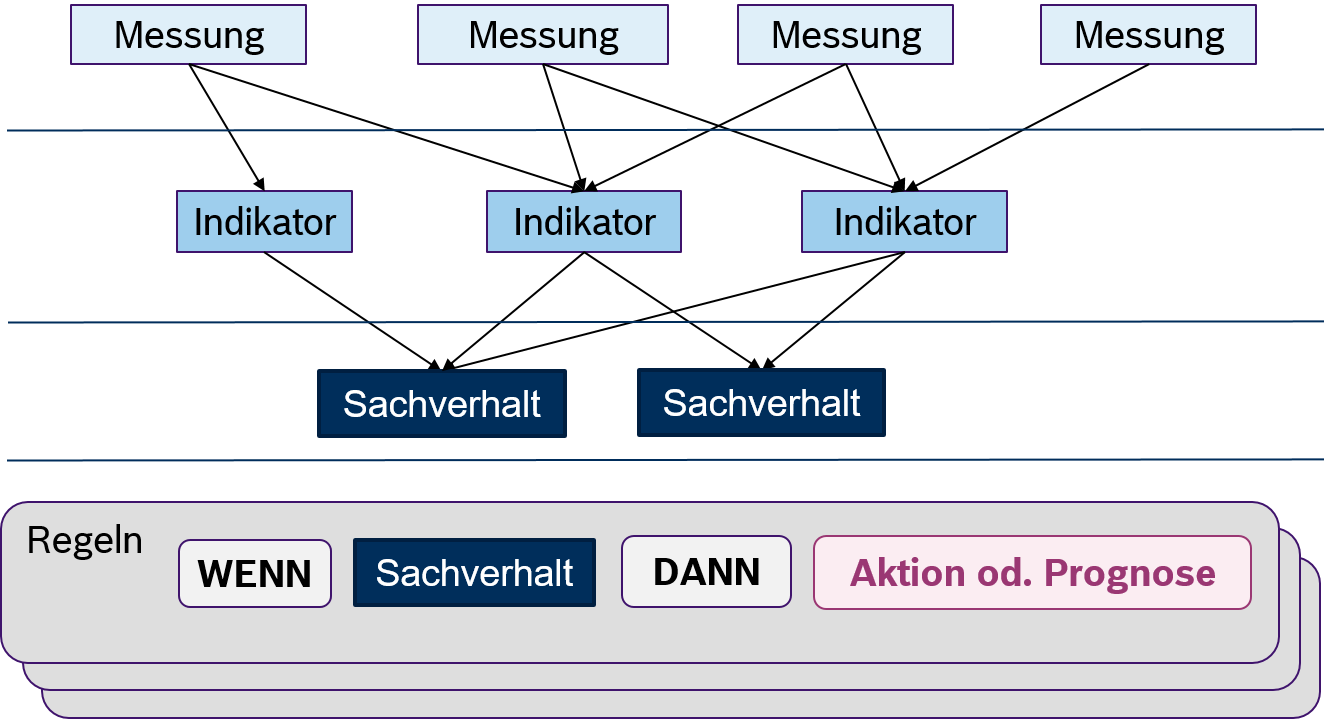}
	\caption[Abbildung]{Sachverhalte und Regeln}
	\label{fig:sachverhalte}
\end{figure}

Regeln, und das ist für die vorliegende Arbeit eine wesentliche Eigenschaft, bezeichen im Allgemeinen nicht nur Aktionen, die nach außen wirken, sondern auch jene, die nach \glqq innen\grqq{} gerichtet sind, also der Informationsverarbeitung oder Prognosenbildung dienen.

Man mag einwenden, dass eine Aktion nicht nur auf der Basis von Wissen und Signalen, sondern auch wesentlich von der Intention oder des Fahrauftrages, ein bestimmtes Ziel zu erreichen,  abhängt.  Da der Fokus der Arbeit auf dem in Abbildung \ref{fig:szenariosmall} dargestellten Szenario liegt wird als Grundintention des Automaten angenommen, dieses gewählte Szenario nach Süd oder nach Ost im Sinne der Darstellung zu verlassen. Dynamische Änderungen von Intentionen sind natürlich denkbar, sie sind verbunden mit Prognosen und hätten die Natur einer Entscheidung. Ist in Fahrtrichtung voraus in der Stadt ein Stau, so könnte die Prognose des Fahrers lauten, ein Umweg führe schneller zum Ziel. Im Interesse der verständlicheren Darstellung wird die Prognosebildung hier nicht näher untersucht.

\section{Ziele}\label{sec:ziele}
Mit dem bisher Beschriebenen ist das Ziel dieser Arbeit nun präzisierbar: Über eine geeignete Formalisierung soll ein Graph erzeugt werden können, der Informationsflüsse und Aktionen eines Verkehrs-Szenarios abbilden kann. Elemente dieser Abbildung sind Straßen, Akteure,  diverse Signalquellen, wie Verkehrsschilder, Lichtanlagen und dergleichen. Damit soll sichtbar werden, welche Aktionen bezogen auf die Informationslage möglich sind. Strukturierung des Problemfeldes und explizite Notation von Annahmen sind Vorteile, die ebenfalls der Intention dieses Ansatzes entsprechen. Über 3 Schritte wird diese Formalisierung in den Aufbau eines Graphen überführt: 

\begin{enumerate}
	\item Formalisierung von Information und deren Transformation zu Sachverhalten mittels Signalen.
	\item Definition von Regeln, die auf diese Sachverhalte formuliert sind.
	\item Anwendung aller Regeln zur Grapherstellung.
\end{enumerate}

Die Formalisierung sollte so erfolgen, dass eine Implementation in Software möglich ist. Der weitere Text ist diesen Schritten gewidmet.

An dieser Stelle ist ein methodischer Hinweis angebracht. Diese Arbeit steht im Geiste des konstruktiv-empirischen Erprobens, d.h. auf der Basis von plausiblen Ideen werden Bausteine geformt und damit Aussagen, Modelle oder Lösung\-en konstruiert mit dem Ziel, wesentlich empirische Gegebenheiten zu rekonstruieren oder beschreibar zu machen. Im Allgemeinen wird dann eine Erklärung oder Begründung  empirischer Phänomene  akzeptiert, wenn eine Konstruktion aus diesen Bausteinen oder ein Softwarelauf  die  Phänomene in den intendierten Aspekten reproduziert. Hier jedoch ist die Beschreibung des Bauskastens erstes Ziel und Grundlage für weitere Arbeiten.

\section {PSM - Elemente}\label{sec:sachverhaltebasis}
Zunächst soll begründet werden, warum der  mittels Abschnitt \ref{sec:ansatz}, Abbildung \ref{fig:sachverhalte} motivierte Ansatz im Zusammenspiel mit den eingführten Begrifflichkeiten sinnvoll und ziel\-führ\-end ist. 

Es sei dafür an die bereits dargelegte Fragestellung erinnert, wie Erfassung von Aspekten der Umwelt, Information, Wissen und Signale im Kontext des Verkehrsgeschehens modelliert werden könnten. Menschliche Eigenschaften und Handlungen im Verkehr, ihre Natur in Bezug auf Information etc. sind also zu erfassen und auf die Maschinenwelt zu übertragen. Eine Antwort kann offensichtlich nicht gefunden werden, ohne die besonderen Unterschiede von Mensch und Maschine zu berücksichtigen. Weder nimmt der Mensch  objektiv wahr,  noch ist der Mensch als soziales Wesen, eine der Grundaussagen der Husserlschen Philosophie, von gesellschaftlichen Manifestationen unabhängig, wie es z.B. gesellschaftliche Normen darstellen. 

Das PSM gliedert sich daher in zwei Teilbetrachtungen: 

\begin{enumerate}
	\item die \emph{intentionale Erfassung} und
	\item das \emph{intentionale Wissen} 
\end{enumerate}

\subsection{Intentionale Erfassung}

Unter dem Begriff \emph{intentionale Erfassung} sei folgende Interpretation der Husserlschen Arbeiten verstanden:  Eine Erfassung im Sinne einer oder mehrerer Messungen ergibt Information. Die Messung von Größe und Bewegung eines Objektes etwa sind durch die Bezogenheit auf das Objekt verbunden und daher eine Information über dieses Objekt. Aber warum sollte dieses Objekt überhaupt von Interesse sein, oder speziell seine Bewegung?

Für den erfahrenen Verkehrsteilnehmer sind solche Informationen von Interesse. Sie haben für ihn Bedeutung, denn seine Intention ist es ja, Kollisionen zu vermeiden oder die Folgen seines Handelns abzuschätzen. Diese Intentionalität bezogen auf die Information, die Bedeutung verleiht,  bezeichnen wir hier als \emph{Signal}. Eine Bewegung von rechts kann ein Signal für einen möglichen Vorfahrtskonflikt sein. Mit anderen Worten: eine erfasste Information wird zum Signal, erhält also Bedeutung für den Empfänger, wenn jene ein Anzeichen für etwas im Vorerleben oder Vorwissen Vorhandenes ist. Feuer ist ein Anzeichen von Gefahr, aber nur, wenn man schlechte Erfahrungen damit hatte, denn ebenso könnte es ein Signal (Anzeichen) für eine Wärmequelle sein. 

Im Straßenverkehr gibt es aber auch bewusst gesetzte Zeichen, wie Verkehrsschilder, Lichtanlagen oder Markierungen auf Fahrbahnen. All diese Dinge sind Beispiele dafür, dass ihre Gestaltung ebenfalls intentional ist.  Für viele Menschen ist z.B. aufgrund ihrer Sozialisation die Farbe Rot ein Signal für Gefahr, eine Tatsache, die sich Warnschilder und Ampeln und sogar Bremslichter zunutze machen. Wenn also ein Sender einer Information, im genannten Beispiel also  der Verkehrszeichenbauer, sich über die Signalwirkung im Klaren ist, also weiß, welche Information wie Form oder Farbe welches Signal für den gemeinen Verkehrsteilnehmer darstellt, so kann dies zielgerichtet, also intentional ausgenutzt werden. 

\emph{Intentionale Erfassung} meint damit, dass jegliche Information im Straßenverkehr auch ein Signal für dessen Empfänger ist, dessen Bedeutung sich in Abhängigkeit des Vorwissens oder Vorerlebens des Empfängers ergibt. 

\subsection{Intentionales Wissen}

Mit dem Begriff \emph{Intentionales Wissen} wird den Arbeiten der Wissenssoziologen Berger und Luckmann Rechnung getragen. In ihrem Buch \cite{BergLuck} untersuchen sie,  wie sich gesellschaftliche Wirklichkeit aus subjektiven Erlebnissen derer Individuen ausbildet. Für das PSM werden diese Ergebnisse wie folgt interpretiert: 
Gesellschaftliche Wirklichkeit wird als gesellschaftliches Wissen aufgefasst, also Wissen, über das jedes Individuum verfügt. Dies ist eine bewusste Vereinfachung im Hinblick auf die beschriebene Zielsetzung hier. 

Aufgrund der oben ausgeführten Subjektivität nimmt notwendiger Weise jedes Individuum der Gesellschaft seine Lebenswelt  - und hier beschränken wir uns ganz auf den Straßenverkehr - unterschiedlich wahr. 

Dazu wieder zwei Beispiele: der Schulweg der Kinder einer Schulklasse mag an dem Punkt, wo die Kinder eine Straße überqueren müssen, gefährlich sein. Während aber der eine Elternteil die Breite der Straße als Signal für das Risiko ansieht, mag die Zahl parkender Autos für andere das Signal für Risiko sein, wieder andere mögen sich an der Verkehrsdichte stören. Aus den unterschiedlichen Wahrnehmungen derselben Situation bleibt aber ein gemeinsamer Sachverhalt bestehen: die Überquerung der Straße an diesem Punkt ist riskant für Kinder.

In einem anderen Falle stelle sich bei den Bewohnern einer Stadt die Meinung ein, bei Passage einer bestimmten Kreuzung sei die Wahrscheinlichkeit eines Unfalls vergleichsweise hoch. Es spielt dabei keine Rolle, ob dieser von den Bewohnern vertretene Sachverhalt entstanden sein möge, weil tatsächlich viele Menschen immer wieder in Unfälle verwickelt wurden, oder ob dieses Schicksal nur einer Person widerfahren war, die aber gut vernetzt entsprechende Meinungsbildung betrieben hat. 

Gesetze können als Objektivierung der Summe individuellen Erlebens aufgefasst werden, denn in deren Tatbeständen sind Situationen oder Gegebenheiten beschrieben, die für alle Mitglieder der Gesellschaft relevant sind  und die sich aus der Entwicklung der Gesellschaft und ihrer Phänomene zu solch objektivem Wissen entwickelt haben. Im Buch \cite{Eng} sind Beispiele beschrieben, wie sich diese Tatbestände auch verändern können, in dem z.B. die Umstände eines Diebstahls sich wandeln, weil aufgrund der Industrialisierung plötzlich auch Strom zu stehlen möglich und reizvoll wurde.

In Anlehnung an die Konstruktion des Rechts soll hier also verabredet sein, dass \emph{Sachverhalte} dieses überindividuelle, also objektive Wissen darstellen, die sich aus den Einzelerfahrungen der Individuen bildet und für alle Individuen relevant ist. Aufgrund der Bedeutung und Herleitung steht das Konzept der  Sachverhalte für das \textit{intentionale Wissen} einer Gesellschaft im PSM. (Für die Rechtskundigen sei angemerkt, dass die Rechtswissenschaften zwischen Tatbestand, Tatbestandsmerkmal und Sachverhalt unterscheiden. Nicht jeder Sachverhalt ist Teil eines Tatbestandes. Diese Unterscheidung soll im Rahmen der PSM fallen gelassen werden. Der Begriff \glqq Tatbestand \grqq{} wird hier nicht verwendet. Soweit es das PSM betrifft ist nur das als  \glqq Sachverhalte \grqq{} bezeichnete regelrelevante, intentionale Wissen von Bedeutung. Dies wird im weiteren Verlauf, insbesondere auch Abschnitt \ref{sec:regeln}, deutlicher).

\subsection{Sachverhalt, Wissen, Signal}
Um die bis dahin entwickelte Begriffswelt nutzbar zu machen ist es geboten, eine technische Perspektive einzunehmen. Intentionales Wissen und intentionales Erfassen sind so noch nicht für eine Maschine umzusetzen. Man muss i.A. annehmen, dass Sachverhalte der Verkehrsgesellschaft als sozialen Gruppe oder dem, was die StVO beschreibt, nicht direkt technisch erfasst oder gemessen werden können. Die Brücke ist also die Frage, wie  Sachverhalte maschinell zu erfassen sind. Hier bietet wieder die Rechtswissenschaft den nächsten Schritt an. Tatbestandsmerkmale und Indizien sind etwas, die, wenn sie festgestellt wurden, auf einen Tatbestand hinweisen. Aber auch in der Physik werden Grössen oft indirekt bestimmt. Ein Thermometer beispielsweise zeigt nicht direkt die Temperatur an, sondern eigentlich die Dichteänderung der Materie aufgrund der Temperatur. Jene Dichteänderung ist also ein Indikator für eben jene Temperaturänderung. 

Diesen Beobachtungen Rechnung tragend seien \emph{Indikatoren} diejenigen Größen, die das Vorliegen eines Sachverhalts definieren. Das Vorhandensein eines Indikators kann wiederum durch eine oder mehrere technisch durchführbare Messungen überprüft werden.

Damit können, zusammen mit Abbildungen \ref{fig:sachverhalte} und \ref{fig:wissenundoperator}, die Kernelemente des PSM nachfolgend definierend beschrieben werden:

\begin{definition}\label{def:psmdef}
	Für Indikator, Signal, Wissen gelte folgende Festlegung:
	
	\begin{itemize}
		\item  \emph{Sachverhalte} bezeichnen das intentionale Wissen, das jedem Teilnehmer der sozialen Verkehrsgesellschaft zur Verfügung steht.
		\item Regeln in dem Automaten beschreiben Aktionen und Kausalitätsvermutungen. Regeln sind auf der Basis von Sachverhalten formuliert (WENN Sachverhalt DANN...). Wissen im Automaten ist die Summe aller Sachverhalte und Regeln.
		\item 
		 \emph{Sachverhalte} bezeichnen auch erfasste und erkannte Gegebenheiten und damit  das Wissen des Automaten über die Umwelt.
		\item 
		Ein \emph{Sachverhalt} sei gegeben, wenn alles dafür notwendige intentionale Erfassen, vermittelt durch \emph{Indikatoren}, die diesen Sachverhalt definieren,  gegeben sind.
		\item Ein \emph{Indikator} ist eine durch eine oder mehrere technisch durchgeführte Messungen beobachtbare Eigenschaft der Umgebung des autonomen Fahrzeugs.  Ein Indikator gelte als gegeben oder evident, wenn alle entsprechenden Messungen vorliegen.
		\item Das Ereignis, dass ein Indikator evident geworden ist, entspricht dem Erfassungsereignis.
		\item \emph{Signal} ist dasjenige Element, welches benötigt wird, um eine Erfassung in einen Sachverhalt zu transformieren.
	\end{itemize}
	
\end{definition}

Hiermit schließt sich der mit Abschnitt \ref{sec:ansatz} und \ref{sec:begriffe} begonnene Kreis.
	Nebenbemerkung: Für die Zielsetzung dieser Arbeit genügt es, Messungen und die dahinter liegenden Mechanismen als gegeben anzusehen und das philosophische Problem von Realität und Positivismus dahinter nicht zu hinterfragen. Diesem Thema haben sich ganze Disziplinen gewidmet. Eine Übersicht ist z.B. in \cite{Brd} zu finden. 

\section{Symbolische Formulierung}\label{sec:math}

\subsection{Mengen und Abbildungen}

Der Formalisierung von Sachverhalten liegt ein mathematisches Modell zugrunde, welches auf Sequenzen basiert. Das Modell ergibt sich anhand des Zusammenhangs von Sachverhalten und Indikatoren. Sachverhalte werden über Indikatoren erkannt, jedoch die Reihenfolge, in der diese Indikatoren erkannt werden mag i.A. von Bedeutung sein. Dies gibt Anlass zu den in diesem Abschnitt beschriebenen Definitionen.

\begin{definition}\label{def:basis}
	Die Menge $C = \{ c_{1}, c_{2}, ...,c_{n}\}$ mit $n\in\mathbb{N^{+}}$ der Elemente $c_{i}$ heißt \emph{kausale Grundmenge}. Die Elemente heißen \emph{Causae}.
	Daneben soll eine Menge $R=\{\varphi_{1}, \varphi_{2}, ...,\varphi_{n} \}$ gegeben sein, deren Elemente $\varphi_{i}$ \emph{Successus} heißen. Schließlich sei folgende Abbildung $\mathbf{F}$ definiert: 
		\[
	\mathbf{F} : R \times C \longmapsto W, \quad \mathbf{F}( \varphi_{r(i)}, c_{i} ) = \varphi_{r(i)}c_{i}, \quad  \varphi_{r(i)} c_{i} \in W
	\]
	Die $r(i)$ sind die  zu $\mathbf{F}$ gehörenden Indexfunktionen für die Zuweisung eines Index der Menge $C$ auf einen Index der Menge $R$. Mit dieser Abbildung wird jeder Causa aus $C$ ein Successus $\varphi$ zugeordnet, die Paare $ \varphi_{r(i)} c_{i}$ heißen  \emph{Effectus}. Entsprechend heißt die Menge $W$  \emph{Effektmenge} der kausalen Grundmenge. 	
\end{definition}

\begin{bemerkung}
	Zur Vereinfachung der Schreibweise soll im Folgenden mit $\varphi_{i}c_{i}$ immer das zugehörige Paar als Ergebnis der Abbildung $\mathbf{F}$ gemeint sein. Es soll nicht anzeigen, dass $r(i) = i$ sei. Wo es von Bedeutung ist, wird die ausführliche Schreibweise angewendet.
\end{bemerkung}
		 
\begin{example}
	Die $c_{i}$ können mit Indikatoren identifiziert werden. Deren Wirkung durch den Successus $\varphi_{i}$ ergibt die Möglichkeit, etwas zu messen. Der Indikator \glqq Ausdehnung\grqq{} erlaubt z.B. die Messung oder Bestimmung von Werten wie  \glqq 2m hoch \grqq{}, \glqq 1m breit\grqq{} etc. Die Menge $C$ enthält dann Indikatoren, deren Effectus messbar ist.
\end{example}
	
Effectus und Causae können aneinandergereiht werden, also Sequenzen bilden, was nachstehende Definition erklärt.

\begin{definition}\label{def:seq}
	
	Eine \emph{Sequenz} $s^{k}$ der Länge $k$ wird mittels der Elemente einer beliebigen, nicht leeren Menge  $A$   wie folgt definiert:
	Mit  $a_{i} \in  A$ und $i, k\in\mathbb{N^{+}}$ sei
	\begin{align*}
		s^{1} &= a_{1} \\
		s^{k} &= a_{1} a_{2}, ...,a_{k}. \\
	\end{align*}
 	Die Menge aller Sequenzen $S = \{s^{k1}_{1}, s^{k2}_{2},...,s^{km}_{m} \}$ mit $i, k, m \in \mathbb{N^{+}}$ heiße \emph{Sequenzmenge}.
\end{definition}

\begin{definition}\label{def:map}
	Für eine Sequenz $s^{k}$ mit $k \in \mathbb{N^{+}}$ sei eine Verknüpfung durch folgende Abbildung erklärt:
	\[
		f : S \times S \longmapsto S, \quad 
	\]
	mit 
	\begin{equation*}
		f(s^{k}_{1}, s^{1}_{2}) =  
		\begin{cases} 
			 s^{k+1} & \text{wenn } s^{k}_{1} \neq s^{1}_{1}  \\
			s^{k}_{1} &  \text{wenn }  s^{k}_{1} = s^{1}_{2}
		\end{cases}
	\end{equation*}
	
Eine Sequenz wird also nur verlängert, wenn das neue Element nicht dem letzten der Sequenz gleicht. Zur Begründung des ersten Falles:
	\begin{align*}
			f(s^{k}, s^{1}) = s^{k} \; s^{1} &=a_{1} a_{2}, ... a_{k} \; a_{1} \\
			&= a_{1} a_{2}, ...a_{l} \;a_{k+1}  \\
			&= s^{k+1}. 
	\end{align*}	
	
\end{definition}

\begin{lemma}
	Seien $s^{k}$ und $s^{l}$ unterschiedliche Sequenzen mit $k,l\in \mathbb{N^{+}}$, dann ist
	\begin{align*}
		s^{k+l} &= s^{k} \: s^{l} \\
		&= s^{1}_{1}  \: s^{1}_{2}  \: ...  \: s^{1}_{k}  \: s^{1}_{k+1}  \: ...  \: s^{1}_{k+l}.
	\end{align*}
\end{lemma}
\begin{proof}
	Nach Definition  \ref{def:map} ist 
	\begin{align*}
		s^{k} &=   s^{k-1}  \: s^{1}_{k}.  
	\end{align*}
	$s^{k-1}$ kann wieder dargestellt werden gemäß Definition \ref{def:map} als
	\begin{align*}
		s^{k-1} &= s^{k-2}  \: s^{1}_{k-1}
	\end{align*}
   und so 
   	\begin{align*}
   s^{k} &=   s^{k-1}  \: s^{1}_{k} \\
   		&=   s^{k-2}  \: s^{1}_{k-1}  \: s^{1}_{k}.
   \end{align*}
   	Nach insgesamt $k$ Schritten erhält man entsprechend
   	\begin{align*}
   	 s^{k} = s^{1}_{1}  \: s^{1}_{2}  \: ...  \: s^{1}_{k}.
   	\end{align*}
   
    Verfährt man mit $s^l$ genauso und setzt beide Darstellungen in die Behauptung ein, ergibt sich nach allfälliger Umbenennung (zur Verdeutlichung ist $s^{l}$ als $s^{*}$ geschrieben)
    
    \begin{align*}
  	 s^{k}  \: s^{*,l} &= s_{1}  \: s_{2}  \: ...  \: s_{k}  \: s^{*}_{1}  \: s^{*}_2  \: ...  \: s^{*}_{l} \\
  	 &=  s_{1}  \: s_{2}  \: ...  \: s_{k}  \: s_{k+1}  \: ...  \: s_{k+l}\\
  	 &= s^{k+l}.
    \end{align*}

\end{proof}

\begin{bemerkung}\label{rem:interpretation}
	Für die Abbildung $f$ sind nun verschiedene Interpretationen möglich:
	\begin{itemize}
		\item Werden die  $a_{i}$ der Definition \ref{def:seq} mit den Elementen aus  $W$ identifiziert, so ist $W = \{ \varphi_{1} c_{1}, \varphi_{2} c_{2}, ...,\varphi_{k} c_{k} \} = \{s^{1}_{1}, s^{1}_{2}, ..., s^{1}_{k}\}$ mit $s^{1}_{i} \in S_{W} \supseteq W$ und  $k\in\mathbb{N^{+}}$ (der untere Index dient wieder nur der Unterscheidung).  
		
		\item Setzt man $a_{i} = c_{i} \in S_{C} \supseteq C$ und $f_{C} : S_{C} \times S_{C} \longmapsto S_{C}, $ so erzeugt die Abbildung $f_{C}$ Sequenzen der Art $s^{k} = c_{1}c_{2}...c_{k}$.
		
		\item Werden die $a_{i}$ mit den Elementen von $R$ gleichgesetzt, so ist ganz entsprechend Obigem $f_{R} : S_{R} \times S_{R} \longmapsto S_{R}$ und  $s^{k} = \varphi_{1}\varphi_{2}...\varphi_{k} $.
		
	\end{itemize}
\end{bemerkung}

\begin{bemerkung}
	
	Zur Vereinfachung der Schreibweise kann der hochgestellte Index zur Längenangabe auch weggelassen werden, d.h. statt $s^{k}_{i}$ für eine Sequenz der Länge $k$ wird einfach $s_{i}$ geschrieben. Mit dieser Bezeichnung ist dann eine Länge $k \in \mathbb{N^{+}}$ impliziert, aber nicht explizit angezeigt.
		
\end{bemerkung}

\begin{definition}
Folgende Abbildung $\mathbf{H}$ definiert: 
\[
\mathbf{H} : S_{R} \times S_{C} \longmapsto \overline{W}, \quad \mathbf{H}( \varphi_{h(i)}, c_{i} ) = \varphi_{h(i)}c_{i}, \quad  \varphi_{h(i)} c_{i} \in \overline{W}
\]
Die $h(i)$ sind wieder die  zu $\mathbf{H}$ gehörende Indexfunktion für die Zuweisung eines Index der Menge $S_{C}$ auf einen Index der Menge $S_{R}$. Mit dieser Abbildung wird jeder Sequenz aus $S_{C}$, welches hier die Bedeutung einer Causa hat,  eine Sequenz aus $S_{R}$, welches ein  Successus darstellt,  zugeordnet.  Entsprechend heißt die Menge $\overline{W}$  \emph{Effektmenge} der kausalen Grundmenge. \par 

\end{definition}

Alle bisher eingeführten Abbildungen lassen sich in folgendem  \emph{Sequenz - Realisierungs - Diagram} visualisieren: 

\begin{figure}[h]
\centering

\begin{tikzcd}[row sep=huge, column sep = huge]	
	R \times C 
		\ar[r, "f_{C}; f_{R}"]
		\ar[d, "\mathbf{F}"]
	 &  S_{R} \times S_{C} \arrow[d, "\mathbf{H}", shift left=2ex] \\
	W 	\arrow[r, "f"]		& S_{W}, \overline{W} \ar[lu, dashrightarrow, shift left = 0.7ex] \ar[lu, dashrightarrow, shift right = 0.7ex]
\end{tikzcd}	

\caption{Sequenz-Realisierungs-Diagramm}
\label{fig:commute}
\end{figure}

Würde  $S_{W} = \overline{W}$ gelten und so das Diagramm kommutieren, so würde das bedeuten, dass die Sequenzbildung unabhängig von dem Weg auf dieselbe Menge von Sequenzen hinausläuft. Die anschließende Definition beschreibt ein Kritierum hierfür. 

\begin{definition}
	
Das Sequenz-Realisierungsdiagramm nach Abbildung \ref{fig:commute} kommutiert, d.h. es gilt  $S_{W} = \overline{W}$, wenn ein Isomorphismus $\mathbf{T} : S_{W} \longrightarrow \overline{W} $ existiert mit

\[
	 s_{1} = \varphi_{1}c_{1}\varphi_{2}c_{2},...,\varphi_{k}c_{k} \in S_{W}
\]	
und - als Ergebnis von $\mathbf{H}$
\[
		\overline{s_{2}} = \varphi_{1}\varphi_{2},...,\varphi_{k}c_{1} c_{2},...,c_{k} \in \overline{W},
\]
dann soll gelten 	
\[
 \mathbf{T} : S_{W} \longrightarrow \overline{W}, \{s_{1} \in S_{W}, \overline{s_{2}} \in \overline{W} \;|\; s_{1}  = \overline{s_{2}}  \}.
\]

\end{definition}

Da in Defintion \ref{def:basis} über die Natur der Elemente der Menge $C$ nichts weiter gefordert wurde, können die Elemente der Menge  $\overline{W}$ als Elemente einer Menge $C^{''}$, einer kausalen Grundmenge 2. Ordnung angesehen werden. Mithin sei $\overline{W} \subseteq C^{''}$. Zusammen mit einer beliebigen Menge $R^{''} = \{\omega_{1},\omega_{2}, ..., \omega_{o} \}$ von Successus zweiter Ordnung und Anwendung von $\mathbf{F}$ ergeben sich Sequenzen der Form
\begin{align}
	\omega_{j} \; \varphi_{1} c_{1} \varphi_{2} c_{2} ... \varphi_{p}c_{p} =&  
	\omega_{j} \; \varphi_{1} \varphi_{2} ... \varphi_{p} c_{1} c_{2}... c_{p} \label{eq:commute}
\end{align}

mit $p,o \in \mathbb{N^{+}}$ und $j \in \{1, 2, ... o\}$. 

\begin{bemerkung} ~
	
	\begin{itemize} 
		\item Es ist leicht zu sehen, dass Abbildung \ref{fig:commute} in allen Elementen jede beliebige Ordnung annehmen kann.
		\item Gleichung \eqref{eq:commute} macht es nun möglich, über Successus von Sequenzen zu sprechen. Da - wie in Bemerkung \ref{rem:interpretation} festgestellt - auch die $\varphi_{i} c_{i} = s^{1}$, also die Effectus erster Ordnung als Sequenzen der Länge 1 verstanden werden können, werden weitere Definitionen und Sätze nur noch über Sequenzen ausgedrückt.
	\end{itemize}
\end{bemerkung}

\subsection{Kalkül}\label{sec:kalkuel}

Aus dem vorigen Abschnitt wurde deutlich, dass beliebige Sequenzen $s^{k}_{i}$ Successus zugewiesen und die Resultate zu Sequenzen höherer Ordnung verkettet werden können wie $s^{'} = \omega_{1} s_{1} \omega_{2} s_{2} ... \omega_{n} s_{n}.$ Die nun folgenden Definitionen kennzeichnen spezielle Eigenschaften der Successus, wie Neutralität oder Konstanz und bereiten die Grundlage eines Kalküls. Zuvor muss jedoch präzisiert werden, was unter \glqq Gleichheit\grqq{} zu verstehen ist.

\begin{definition}\label{def:gleich}
	
Seien $s_{i}, s_{k}$ beliebige Sequenzen mit $ s_{i}, s_{k} \in \overline{W}$. Die Sequenzen heißen \emph{E-gleich}, wenn mit einem beliebigen Successus $n+1$.ter Ordnung $\omega \in R^{(n+1)} $ gilt
\[
	s_{i} \stackrel{E}{=} s_{k} \Longleftrightarrow \omega\:s_{i} = \omega\:s_{k}
\]
	
\end{definition}

Unter diesen Voraussetzungen lassen sich jetzt spezielle Sequenzen angeben.

\begin{definition}
	Ein Successus $\varphi^{c} \in R^{(n+1)}$ heißt \emph{Konstanz}, wenn gilt:
	\[
	\varphi^{c} s_{i} = \varphi^{c} s_{j} \quad \text{ für alle} \quad s_{i}, s_{j} \in \overline{W}.
	\]
\end{definition}

Eine Konstanz verhindert also, dass die durch die Sequenz gegebenen Causae einen Effekt haben. 

\begin{definition}
	Ein Successus n-ter Ordnung $ \mathcal{I}$ heißt \emph{Neutrum}, wenn gilt:
	\[
	\mathcal{I} s_{i} \stackrel{E}{=} s_{i} \quad \text{ für alle} \quad s_{i} \in \overline{W}.
	\]
\end{definition}

\begin{definition}
	Ein Successus n-ter Ordnung $\omega^{-1} c$ heißt \emph{Inverseffekt} oder \emph{Neutralisator}, wenn  gilt:
	\begin{align*}
		 \:\omega s_{i} \omega^{-1} s_{j} & \stackrel{E}{=}  \omega \omega^{-1} s_{i} s_{j}  \\ & \stackrel{E}{=} 
		  \:\mathcal{I}s_{i} s_{j} \quad \text{	für alle} \quad  s_{i}, s_{j} \in \overline{W}.
	\end{align*}
\end{definition}

\begin{lemma}\label{lem:nodouble1}
	Sei $\omega s_{1} \omega s_{2} $, dann  gilt 
	
	\[
	 \omega s_{1} \omega s_{2}  \stackrel{E}{=} \omega s_{1}s_{2}.
	\]
	
\begin{proof}
	Da das Diagramm \ref{fig:commute} kommutiert, kann die Sequenz $\omega s_{1} \omega s_{2} $ auch mit der Abbildung $f$ aus \ref{def:basis} dargestellt werden als 
	\[
		f(\omega, s_{1}) f(\omega, s_{2}) \stackrel{E}{=} f(\omega, \omega) f(s_{1}, s_{2}) = \omega s_{1} s_{2}
	\] 
	aufgrund der Eigenschaft von $f$, Fall 2.
	
\end{proof}

\end{lemma}\label{lem:nodouble2}
Es gibt also keine Doppelung gleicher Effekte innerhalb einer Sequenz. Mit gleicher Beweisstruktur folgt

\begin{lemma}
	Sei $\omega_{1} s \omega_{2} s$, dann  gilt 
	
	\[
	\omega_{1} s \; \omega_{2} s  \stackrel{E}{=} \omega_{1} \omega_{2} \;s.
	\]
\end{lemma}

Unter bestimmten Bedingungen kann die Reihenfolge der Causae getauscht werden. Dies beschreibt folgendes Lemma.

\begin{lemma}
	Seien $\omega s_{i}$ und $\omega s_{j}$ zwei Effectus aus $\overline{W}$ und es gelte
	\[
		\omega s_{i} \omega s_{j} \stackrel{E}{=} \omega s_{j} \omega s_{i} 
	\] dann ist
	\[
			 s_{i} s_{j}  \stackrel{E}{=}  s_{j} s_{i}
	\]
	
\begin{proof}

	\begin{align*}
 \omega s_{i} \omega s_{j} &\stackrel{E}{=} \omega s_{j} \omega s_{i} \\
	\omega\:\omega s_{i} s_{j}  &\stackrel{E}{=}	\omega\:\omega s_{j} s_{i} \\
	\omega s_{i} s_{j}  &\stackrel{E}{=}	\omega s_{j} s_{i} \\
 	s_{i} s_{j} &\stackrel{E}{=} s_{j} s_{i}
	\end{align*}

\end{proof}
\end{lemma}

\section{PSM Symbolik der Sachverhalte}\label{sec:symbolik}

Nun stellt sich die Frage, welche Messungen und welche Indikatoren sich zu Sachverhalten ergeben. Messungen sind stark von technischen Randbedingung beeinflusst. Arbeitsgruppen im VVM beschäftigen sich mit entsprechenden Sensormodellen. Um die Diskussion nachvollziehbar zu halten, werden Indikatoren nur abstrakt definiert.  Eines sei aber doch über die Natur dieser Beziehung vermerkt: Im Allgemeinen stecken in der Frage, welche Indikatoren einen Sachverhalt als evident erscheinen lassen, die Kausalitätsvermutungen, die uns Wissenschaft und Erfahrung vermitteln. Eine empirische Wissenschaft stellt auf der Basis von Beobachtungen Hypothesen über Zusammenhänge auf, die sich ausnutzen lassen, um Indikatoren zu finden. Im Falle des bereits genannten Thermometers ist die temperaturabhängige Ausdehnung einer Flüssigkeit eine solche Kausalitätsvermutung und so eine mögliche Option für einen Indikator für die Temperaturbestimmung. 

Normative Wissenschaften wie das Recht erklären solche Kausalitäts\-vermu\-tungen, oder sie bilden sich als gesellschaftlicher Konsens heraus \cite{BergLuck}.  \glqq Diebstahl ist, wenn jemand Eigentum eines anderen entwendet\grqq{} wäre eine solche, hier zur Illustration stark vereinfachte,  gesetzte Kausalität.  Wichtig ist nur, dass die Relation zwischen Sachverhalten, deren Indikatoren und Messungen eine modellhafte Abbildung dessen sind, was in der Lebenswelt der Gesellschaft als bedeutsam bzw. intentional angesehen wird. Die Wahl von Indikatoren und Sachverhalten ist somit einerseits andwendungsabhängig, andererseits ermöglicht dies aber auch, die Aussagen der PSM auf die Fragestellung fein abzustimmen. Soweit es die vorliegende Arbeit angeht, werden Indikatoren beispielhaft und ohne tiefere Begründung ausgeführt. Für Details zur Wahl von Sachverhalten und Indikatoren siehe \cite{Salem}.

Nachstehende Tabelle \ref{tab:indikatoren} zeigt eine beispielhafte Auswahl von Indikatoren, deren Symbole und Semantik. Diese Auswahl wurde im Hinlick auf Abbildung \ref{fig:szenario} und mit der Absicht getroffen, die grundlegenden Ideen verständlich darzustellen, so dass diese für umfangreiche Anwendungsfälle leicht adaptiert werden kann.

\begin{table}[h]
	\caption[Tabelle]{Indikatoren}
	\small
	\label{tab:indikatoren}
\begin{tabularx}{\textwidth}{|c|c|l|X|}
	\hline 
	\textbf{Symbol} & \textbf{Indikator} & \textbf{Wertebereich} & \textbf{Erläuterung} \\ 
	 &  Causa &  Successus &\\ 
	\hline 
	P & Position & b1, b2, g1, r1 & Verweis auf die Position einer Erfassung. Die Angabe erfolgt in Zonen nach Abbildung \ref{fig:szenario} \\ 
	\hline 
	A& Ausdehnung & r, f, a, l &  Ausdehnung einer Erfassung, also Ausdehnung eines Fußgängers f, Radfahrers r, Autos a oder LKW l \\ 
	\hline 
	Q& Qualität & ü, r, f, a, l &  entsprechend A, jedoch als eine abstrakte Qualität, z.B. Reflexionseigenschaften o.ä. \glqq ü\grqq{} steht für Überweg (Zebrastreifen)\\
	\hline
	R& Richtung & $<$, $>$, +, -& Richtung der Erfassung gemessen vom Ego Fahrzeug,"+" für gerade aus in Fahrtrichtung,  \glqq$<$\grqq{} links, \glqq$>$\grqq{} rechts  \\ 
	\hline 
	B& Bewegung & 0, $<$, $>$, +, -& Bewegung in entsprechender Richtung. 0 bedeutet keine Bewegung oder Stop   \\ 
	\hline 
\end{tabularx} 

\end{table}
\paragraph{}

Diese Indikatoren werden nun in den symbolischen Formalismus nach Abschnitt  \ref{sec:math} eingebettet, indem sie die Menge $C$ aus Definition \ref{def:basis} bilden. Deren Wertebereich ist der Wertebereich der Successus $R$ gemäß eben dieser Definition. Damit ergibt sich:

\begin{align*}
C &= \{P, A, Q, R, B\} \\
R &= \{b1, b2, g1, r1, r, f, a, l, <, >, + , -, 0\}
\end{align*} 

Tabelle \ref{tab:indikatoren} in mathematischer Schreibweise zu notieren würde ergeben:  $b1P \in W, rA \in W, rQ \in W$, jedoch z.B. $<P = \mathcal{I}$, d.h. alle Paare $\varphi_{r(i)} c_{i}$, die nicht durch den in besagter Tabelle abgebildeten Wertebereich abgedeckt sind, realisieren nur das Neutrum $ \mathcal{I}$. Mittels dieser Interpretation ist Tabelle \ref{tab:indikatoren}  vollständig als Menge $W$ definierbar. 

Mit Bezug auf das Sequenz-Realisierungs-Diagramm \ref{fig:commute} können durch Verkettung mit $f$ Effektussequenzen aufgebaut werden. Der Übergang in die 2. Ordnung, also Successus 2. Ordnung dieser Effektussequenzen beschreibt Erfassung und Sachverhalt mittels folgender Definition: 

\begin{definition}
Sei gemäß oben beschriebenen 
\[
C^{''} := W = \{b1P, b2P, ..., rA...rQ..., +R..,\mathcal{I}\}
\]

die Causae 2. Ordnung (= $\varphi_{f(i)} s_{i}$), sowie
\[
R^{''} = \{!, ?, !^{-1}, ?^{-1}\}
\] 

die Successus 2. Ordnung, inklusive deren Inverseffekte. Dann heißen die Bilder  $\mathbf{F}(\varphi,s) = \:?\:\varphi_{f(i)}s_{i} $ die \emph{Erfassungsereignisse} und die Bilder  $\mathbf{F}(\varphi,s) = \:!\:\varphi_{f(i)}s_{i} $ die \emph{Sachverhalte} eines Szenarios.
\end{definition}

Es ist klar, dass die Menge aller Effectus-Sequenzen  $\varphi_{f(i)} s_{i}$, die die Natur von Sequenzen realisierter Indikatoren hat,  alle überhaupt darstellbaren Sachverhalte und Erfassungen im Modell bestimmt. Sachverhalte stellen sich somit als ein Successus 2. Ordnung mit einem Causa dar, formal $!\:\varphi_{f(i)}s_{i} $ mit $i=1...k$ und $k \in \mathbb{N^{+}}$, wobei der Causa aus einer Effectus-Sequenz $\varphi_{f(i)} s_{i}$ der Länge $k$ besteht. 

Durch die Wahl der Indikatoren und ihrer Realisierungen (der Successus) kann somit das Modell ganz auf die Problemstellung eingestellt werden. Nun ist der Weg frei, Wissen und Sachverhalte in einer Maschine darzustellen und unter Anwendung des Kalküls nach Abschnitt \ref{sec:kalkuel} Regeln formalisiert darauf auszudrücken. 

\section{PSM -  Regeln}\label{sec:regeln}

Mit den in einer Symbolsprache codierten Sachverhalten können jetzt Regeln entsprechend Abschnitt \ref{sec:ziele} kompakt und einfach notiert werden.  Regeln beschreiben unterschiedliche Aspekte des Gesamtmodells. Um diese zu klären, sei nochmals auf Abbildung \ref{fig:wissenundoperator} verwiesen. Alle dort gezeigten mit dem Symbol $\alpha A$ bezeichneten Prozesse sind Gegenstand von Regeln. Entsprechend deren Rolle lassen sich folgende Regeltypen unterscheiden:

\begin{figure}[h]
	\centering
	\includegraphics[width=0.7\linewidth]{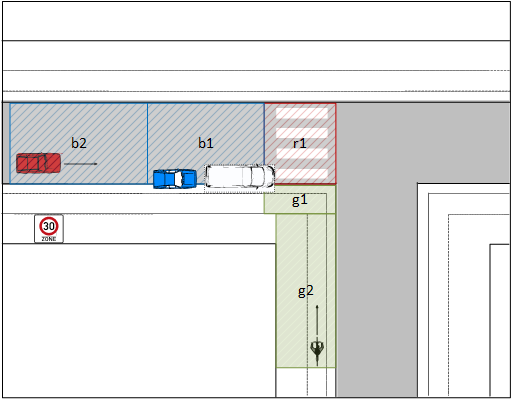}
	\caption[Abbildung]{Beispiel-Szenario }
	\label{fig:szenario}
\end{figure}

\begin{itemize}
	\item \emph{Verhaltensregeln}: Wie in Abbildung \ref{fig:wissenundoperator} zu sehen, führen Wissen und Prognosen zu Aktionen. Sachverhalte repräsentieren im hier vertretenen Modell das Wissen. Regeln, die aus einem Sachverhalt folgende, nach außen sichtbare Aktion beschreiben, sind also Verhaltenstregeln. 
	
	Verhaltensregeln modellieren einerseits das \emph{Normverhalten}, also Erwartungen an Verhalten, die sich auf den im gesellschaftlichen Leben geformten Sachverhalten begründen. Ebenso stehen sie aber für das \emph{Sollverhalten}, welches dasjenige Verhalten bezeichnet, das der Erbauer des autonomen Fahrzeugs als Produkentscheidung haben möchte, oder das als Anforderung aus Gesetzen oder Normverhalten in die Entwicklung eingeht.  Die semantische Analyse der PSM (\cite{Salem}) ist eine Methode,  solche Verhaltensregeln zu identifizieren. 
	
	\item \emph{Strukturregeln}: Die  $\alpha A$ in Abbildung \ref{fig:wissenundoperator}  repräsentieren auch die Erkennung eines Signals als solches. Wie bereits ausführlich beschrieben, ist dieser Vorgang von dem vorhandenen Wissen, gegeben in Form von Regeln und Sachverhalten, abhängig. Nach Definition \ref{def:psmdef} bedeutet ein Signal zu beschreiben, den Weg von der Erfassung zum Sachverhalt zu notieren bzw. zu formalisieren. 
	
	Die Modellierung des Weges von Erfassung zu Indikatoren fällt leicht, denn hier werden meist physikalische Gesetze eine Rolle spielen. Man erinnere sich an das oben erwähnte Beispiel des Thermometers und auch an Abbildung \ref{fig:sachverhalte}. Strukturregeln an dieser Stelle des Modells können damit aus den in Abschnitt \ref{sec:symbolik} erwähnten Kausalvermutungen gewonnen werden.
	
	Weniger offensichtlich ist der Weg von Indikatoren zu erkannten Sachverhalten. Im Husserlschen Sinne transportieren Signale Anzeichen, die,  um nun in den Sprachgebrauch des PSM  überzugehen, für einen oder mehrere bereits in der Maschine vorhande Regeln und deren Sachverhalte stehen. 
	
	Ein Sachverhalt ist ja eine Sequenz (von Effectus). Sind alle Indikatoren realisiert, so ist der damit definierte Sachverhalt erfasst, also  $?\:\varphi_{f(i)}s_{i} $. Findet sich diese Sequenz $\varphi_{f(i)}s_{i}$ als Teil eines hinterlegten Sachverhaltes der Form $?^{-1}\:\varphi_{f(i)}s_{i}$,  so wird nach den Regeln des Formalismus $?\:\varphi_{f(i)}s_{i} $ zu $!\:\varphi_{f(i)}s_{i} $ und der Sachverhalt ist erkannt.
	
	Genau diese Verknüpfung wird ebenfalls mittels Strukturregeln ausgedrückt und entspricht, weil sie den Pfad von Erfassung zu Sachverhalt vervollständigt, zusammen mit den Ausdrücken der Art $?^{-1}\:\varphi_{f(i)}s_{i}$,  der Repräsentation eines Signals im PSM. 
	
	Strukturreglen beschreiben zudem strukturelle Entwicklungen einer Szene, wie Sichtbarkeitsbedingungen (wer kann wen sehen und von wo) oder Bewegungen von Objekten oder Agenten.
	
	\item \emph{Äquivalenzregeln}: Diese Regeln beschreiben Transformationen der (Indi\-kator-) Sequenzen, die eine Gesamtwirkung nicht ändern. In der Sprache des Formalismus entspricht dies dem Konzept des \glqq E-gleich\grqq{} nach Definition \ref{def:gleich}. Ein Beispiel dazu wird dies später illustrieren. Regeln dieser Klasse sind z.B. nützlich, wenn eine Messung durch eine technisch leichter zugängliche Messung ersetzt werden kann.
	
\end{itemize}

Um die angeführten Regeln konkret zu formalisieren, sei nochmals das zugrunde liegene Beispielszenario angeführt. Vorab sei noch angemerkt, dass 

\begin{itemize}
	\item es hier nicht um geometrische Exaktheit in den Regelaussagen geht, vielmehr soll die Methode und der Modellaufbau veranschaulicht werden,
	\item die Regeln, insbesondere die Verhaltensregeln,  immer aus der Sicht des Ego-Fahrzeugs zu formulieren sind, andere Agenten gehen lediglich über Strukturregeln ein,
	\item  große Buchstaben X, Y, Z in den symbolischen Sequenzen für beliebige Indikatoren stehen, 
	\item kleine Buchstaben x, y, z für deren Successus stehen,
	\item die angegebenen Regeln Anwendungen des Kalküls repräsentieren. Mehrere Bedingungen einer Regel sind mit logischem UND verknüpft.
\end{itemize}

\begin{table}[h]
	\caption{Strukturregeln bezogen auf Dynamik}
	\small
	\begin{tabular}{p{0.15\textwidth}m{0.3\textwidth}m{0.5\textwidth}}
		\hline
		Thema & Regel & Bedeutung\\
		\hline
	 	Sichtbarkeit&
	 		\begin{tikzcd}[row sep=0]
	 			Xb1P \ar[rd]   & \\
	 		 & ?Yr1P\\
	 		Yr1P \ar[ru]	 &
	 	\end{tikzcd} & Von Zone b1 aus kann etwas in r1 erfasst werden, sofern dort etwas ist.\\
 	\hline
		Sichtbarkeit & 	\begin{tikzcd}[row sep=0.1]
			Xb2P \ar[rd]    & \\
			 & ?Yr1P\\
			Yr1P \ar[ru] &
		\end{tikzcd}  & Entsprechend für Zone b2\\
	\hline
		Sichtbarkeit & \begin{tikzcd}[row sep=0]
			Xb2P \ar[rd]    & \\
		 & ?Yg2P\\
			Yg2P \ar[ru] &
		\end{tikzcd}  & Von Zone b2 aus kann etwas in Zone g2 erfasst werden.\\
	\hline
		Bewegung & +Bb2P $\longrightarrow$ +Bb1P & Wenn sich ein Agent bewegt, wechselt seine Position von b2 nach b1 und erhält dabei die Bewegung\\
		\hline
		Bewegung & +Bb1P $\longrightarrow$ +Br1P & Entsprechend b1 nach r1 \\
		\hline
		Bewegung & X+Bg2P $\longrightarrow$ X+Bg1P & Ein Objekt in g2 bewegt sich nach g1 (=Radfahrer hier im Beispiel)\\
		\hline
		Bewegung & +Bg1P $\longrightarrow$ +Br1P & g1 nach r1\\
		\hline
		\\
		\end{tabular} 
	\label{tab:structure1}

\end{table}

Bild \ref{fig:szenario} zeigt, dass das Ego Fahrzeug (rot) den Radfahrer nur \glqq sehen\grqq{} kann, wenn dieses in Zone b2 steht und der Radfahrer sich in Zone g2 befindet. Weiter sieht das  Ego Fahrzeug von b2 und b1 aus auch Zone r1, also den Zebrastreifen. Aussagen dieser Art müssen sich in Strukturregeln nach Tabelle \ref{tab:structure1} niederlegen.

\begin{table}[H]
	\caption{Strukturregeln bezogen auf Sachverhalte}
	\small
	\begin{tabular}{p{0.15\textwidth}m{0.3\textwidth}m{0.5\textwidth}}
		\hline
		Thema & Regel & Bedeutung\\
		\hline
		Signal  & 
		\begin{tikzcd}[row sep=0]
			?s_{i} \ar[rd] &    \\
			& !s_{i}\\
			?^{-1}s_{i}!s_{i} \ar[ru]	 &
		\end{tikzcd} & denn mit Lemma \ref{lem:nodouble1} und \ref{lem:nodouble2} ist
		{\begin{align*} 
			&	?s_{i} \; ?^{-1} s_{i} \; ! \; s_{i}  \\
			&\stackrel{E}{=} \;	? \; ?^{-1} \; !\; s_{i} s_{i} s_{i} \\ &\stackrel{E}{=} \mathcal{I} \;! \;s_{i} & 
			\end{align*} }
		Aus der Erfassung wird der Sachverhalt mittels Signal $	?^{-1}s_{i}!s_{i}$.  \\
		\hline
	\end{tabular} 
	\label{tab:structure2}
\end{table}

Die Gesamtheit dieser Regeln eröffnet keine Option, dass der Radfahrer vom Ego Fahrzeug bemerkt wird, wenn dieser sich in Zone g1 befindet. Eine aus einer bestimmten Zone nicht einsehbare andere Zone ist ein Beispiel für ein Kritikalitätsphänomen der Verdeckung. Dies wird hier nicht weiter vertieft, es zeigt aber, dass über die Modellierung von Kausalitätsvermutungen über Regeln auch unfallbezogene Phänomene oder Erkenntnisse von Gefährdungsanalysen abgebildet werden können.  Für weitere Phänomene dieser Art und Details siehe auch Arbeiten aus dem Teilprojekt 2 des VVM \cite{Neu}.  

\begin{table}[h]
	\small
	\caption{Verhaltensregeln}
	\begin{tabular}{p{0.15\textwidth}m{0.3\textwidth}m{0.5\textwidth}}
		\hline
		Thema & Regel & Bedeutung \\ \hline		
		StVO Regel & 	
		\begin{tikzcd}[row sep=0, column sep=small]
			!+B	\ddot{u}Q+R \ar[rd] &   & \\
			& 0B	\\
			!rQg1P \ar[ru]	 &
		\end{tikzcd} 
		& Diese Sachverhalte sind gegeben: Zebrastreifen voraus, Ego vor Zebrastreifen und fährt, Radfahrer an Zebrastreifen, also halten\\ \hline
		StVO Regel & 	
		\begin{tikzcd}[row sep=0, column sep=small]
			!+Bb1P	\ddot{u}Qr1P \ar[rd] &   & \\
			& 0B	\\
			!rQr1P \ar[ru]	 &
		\end{tikzcd} 
		& Letztlich gleicher Sachverhalt wie oben, alternative Darstellung, vgl. Alternativregeln. Im Beispielgraph nicht verwendet.\\ \hline
		Kollision &
		\begin{tikzcd}[row sep=0, column sep=small]
			XxP \ar[rd] &   & \\
			& 00	\\
			YxP \ar[ru]	 &
		\end{tikzcd}  &Dinge in der selben Zone kollidieren (gleiche Successus von P)  \\ \hline
	\end{tabular} 
\label{tab:verhalten}
\end{table}

Die Tabelle \ref{tab:structure2} enthält die Regeln, die Signale betreffen. Verhaltensregeln für das hier behandelte Beispiel werden aus der StVO entnommen und etwas simplifiziert. Im Wesentlichen soll das Ego Fahrzeug anhalten, wenn es am Zebrastreifen oder auf dem Zebrastreifen den Radfahrer erkennt. Für die Formulierung der Regeln sind folgende Punkte zu beachten:

\begin{itemize}
	\item Verhalten bedeutet, eine Aktion ist notwendig, um den bestehenden Zustand oder ein bestehendes Gleichgewicht zwischen Automat und Umgebung zu erhalten. Fahren auf freier Straße wäre ein Beispiel für solche Gleichgewichte. Dies erfordert zwar Gas geben oder das Betätigen der Fahrradpedale, ist aber im PSM keine Aktion! Ebenfalls ist eine Aktion dann notwendig, wenn nur dadurch die Intention des Ego Fahrzeugs, das Fahrziel z.B.,  erhalten werden kann. 
	\item Verhaltensregeln sind Verhaltensregeln des Ego-Fahrzeugs, Agenten gehen mittels Strukturregeln ein.
\end{itemize}

Auch für dieses Beispiel lässt sich eine Äquivalenzregel formulieren, die eine namentliche Benennung der nächsten Zone durch \glqq Zone voraus\grqq{} erstetzt. Sie sind in Tabelle \ref{tab:equivalent} aufgeführt.

\begin{table}[h]
	\caption{Äquivalenzregeln}
	\small
	\begin{tabular}{p{0.15\textwidth}m{0.3\textwidth}m{0.5\textwidth}}
		\hline
		Thema & Regel & Bedeutung \\ \hline		
		Voraus& !b1PxQr1P $\longrightarrow$ !xQ+R & Wenn Ego Fahrzeug in b1 steht und etwas in r1 gegeben ist, ist es gleichbedeutend zu dem Sachverhalt Objekt xQ voraus in Fahrtrichtung.
	\\ \hline
	\end{tabular} 
\label{tab:equivalent}
\end{table}

\section{PSM-Graph}
Mit den im vorhergehenden Abschnitt vorgestellten Regeln lässt sich nun ein Graph für Abbildung \ref{fig:szenario} erstellen. Es sei daran erinnert, dass die Rolle von Information für das Agieren untersucht werden soll. Dies ist eine Frage der Struktur von Information und deren Flüsse. Mit diesem Anspruch im Blick und den erarbeiteten Formalismen kann nun der Graph nach folgendem, prinzipiell maschinentauglichem Algorithmus erstellt werden:

\begin{enumerate}

\item Platziere die strukturellen Gegebenheiten für jeden Agenten als einen Startknoten in den Graph. Im Beispiel sind das $\ddot uQr1P = $ Zebrastreifen in r1, $+Bb2P = $ Ego Fahrzeug in b2 und fahrend, $rQg2P = $ Radfahrer in g2.
\item Erstelle den nächsten Knoten, in dem alle Regeln angewendet werden, die anwendbar sind. Die Regel-Folge gibt den Knoten an, der neu entsteht. Im Beispiel würde nach einer Strukturregel aus $\ddot uQr1P$ und $+Bb2P$ der Knoten $?Yr1P$ mit $Y=\ddot uQ$ folgen usw.
\item Wiederhole Schritt 2, bis keine der verfügbaren Regeln mehr anwendbar sind.
\item Ggf. lösche alle Pfade, die nicht auf einen Aktionsknoten führen.
	
\end{enumerate}

Die Farben wurden zur Illustration manuell hinzugefügt. Grau sind strukturelle Gegebenheiten. Hellblau zeigt die Knoten von Erfassungsereignissen an, Grün die von Sachverhalten. Orange entspricht den Signalen, deren Signatur im Vorwissen vorhanden sein muss. Zu beachten ist, dass Regeln mit mehreren Bedingungen nur dann anzuwenden sind, wenn alle Bedingungen erfüllt sind (UND Beziehung).

\begin{figure} [H]
	\centering
	\includegraphics[width=1.0\linewidth,height=0.6\textheight]{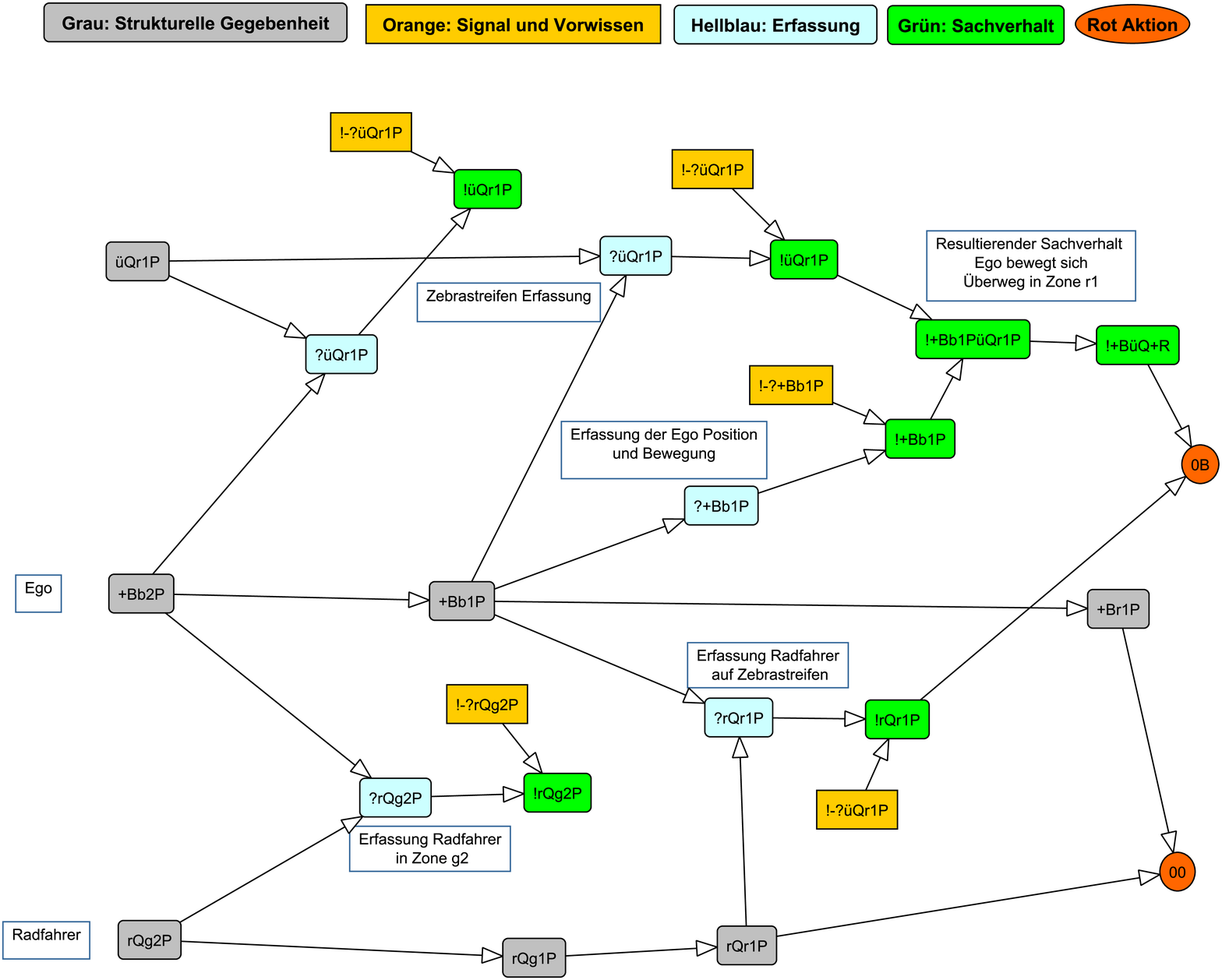}
	\caption[Abbildung]{Beispiel-PSM-Graph}
	\label{fig:psmgraph}
\end{figure}

Im Ergebnis zeigt dieser Graph verschiedene Wege von der Startbedingung des Szenarios zu den Aktionen, wobei unterschiedliche Pfade unterschiedliche Informations- und Signalflüsse repräsentieren. Verschiedene Auswertungen können nun an dem PSM Graphen erfolgen, indem man Pfade durch den Graphen betrachtet. Diese Pfade haben keine unmittelbare Korrelation zu Zeit, aber dennoch sind es Sequenzen und stellen damit eine Abfolge dar. Mit anderen Worten, der Graph symbolisiert die Entwicklung von informationellen Strukturen des inhärent zeitlichen Szenarios.

Der Beispielgraph nach Abbildung \ref{fig:psmgraph} zeigt, dass es Pfade von den grauen Startknoten zum Kollisionsknoten (00) gibt. Und wie aus der Erfahrung zu erwarten ist, führt Nicht-Erfassen in ein unerwünschtes Ereignis: beide Pfade führen nur über graue Knoten, kein Knoten der Erfassung wird durchlaufen. Im Gegenzug führen alle Pfade, die zu einer Bremsung des Ego Fahrzeugs führen (0B), über Erfassungsknoten und Sachverhalte. 

Als Beispiel sei der Pfad der Erkennung des Radfahrers betrachtet. Für das Ego-Fahrzeug wäre dies [+Bb2P, +Bb1P], für den Radfahrer [rQg2P, rQg1P, rQ1P]. Wenn also das Ego-Fahrzeug in b1 steht, der Radfahrer seinerseits in r1, also auf dem Zebrastreifen, dann kann ihn das Ego Fahrzeug erfassen, und über die Sachverhaltsbildung \glqq Radfahrer auf Zebrastreifen\grqq{} (grüner Knoten) die Bremsung einleiten. 

Obwohl nur ein kleines Beispiel behandelt und ein kleiner Auszug gezeichnet wurde, lassen sich hier einige Aussagen ableiten:
\begin{itemize}
	\item Eine Bremsung erfordert einen Sachverhalt. Damit das Ego-Fahrzeug diesen bilden kann, braucht es das Signalverständnis (der orangefarbene Knoten repräsentiert das notwendige Vorwissen dazu). 
	\item Daraus folgt, das Auto braucht also die \textit{Fähigkeit}, u.a. die Erfassung ?rQr1P zu leisten (siehe Pfad) und daraus den Sachverhalt !rQr1P zu bilden. 
	\item Wenn eine Kollision vermieden werden soll, braucht es weitere Regeln (oder modifizierte), damit der Knoten (00) nicht mehr erreicht werden kann.
	\item Sollte die Fähigkeit, den Sachverhalt zu bilden, eingespart und dem Ego Fahrzeug eine von Indikatoren unabhängige Bremsung mitgegeben werden, so wären Regeln zu formulieren, die von der Erfassung direkt auf die Bremsung (0B) führen.
	\item Ein \textit{Sollverhalten} wird repräsentiert durch genau diejenigen Pfade im Graphen des Ego-Fahrzeugs, die zu realisieren sind, beschrieben durch die bezogenen Verhaltensregeln und Fähigkeiten. Letztere sind durch die Erfassungs- und Sachverhaltsknoten (inkl. der Signale) in diesem Pfad angezeigt.
\end{itemize}

Hier zeigt sich also die Natur des PSM Graphen und der Formalismen zu seiner Erstellung als Werkzeug. Durch diesen Graphen wird das Problemfeld strukturiert und einer Bewertung und zielgerichteten genaueren Betrachtung zugänglich. Zusätzlich erlaubt die Formalisierung von Indikatoren und Sachverhalten die transparente Begründung von Relevanz und Sinnhaltigkeit der modellierten Elemente und Annahmen und unterstützt damit eine sicherheitsbezogene Argumentation. Dies folgt aus dem hier beschriebenen Weg der Ableitung der Modellkonstruktue (Indikator, Sachverhalt ect.) bis hin zum Graph. Die Entscheidungen entlang dieser Ableitung sind stets begründbar.

\section{Zusammenfassung und Ausblick}

Mit dem PSM-Konzept werden die Entscheidungsmöglichkeiten eines L4/L5-automatisierten Fahrzeugs abhängig von situationsbedingt empfangenen Signalen über den umgebenden Verkehrsraum und das Verhalten anderer Verkehrsteilnehmer dargestellt.  Von den jeweiligen regelbasierten Entscheidungs\-möglich\-keiten abhängige Veränderungen der Verkehrssituation werden daraus gefolgert. Zudem werden die dem Fahrzeug jeweils nachfolgend  vorliegenden Informationen und daraus abgeleiteten Entscheidungsmöglichkeiten überprüft und der Graph für diese weiterentwickelt. Der Weg dahin führt über Formalisierung von Verhalten, Wissen und Signalen. Das Konzept eröffnet Möglichkeiten zur Darstellung und Auswertung des Verhaltens des Fahrzeugs und daraus resultierender Szenarienverläufe. Dies soll in Folgearbeiten für die Entwicklung des gewünschten Fahrzeugverhaltens, sowie dessen Absicherung und Verifikation genutzt werden.

So wird im weiteren Projektverlauf in VVM, wie zuvor bereits erwähnt, an einer computergestützten Generierung und Auswertung der PSM Graphen gearbeitet. Dies ist für eine praxisgerechte Anwendung des Konzepts notwendig, um die in der Realität gegebene große Zahl von Verlaufs- und Variationsmöglichkeiten des Szenarios bewältigen zu können. Von einer solchen Möglich\-keit ausgehend werden begleitend Fragestellungen für dessen Nutzungs\-möglich\-keiten für die Systemdefinition und eine Verifikation der hinreichenden Sicherheit des implementierten Fahrzeugverhaltens im Verkehrsgeschehen untersucht. 

Unter anderem spielen dabei Fragestellungen der methodologisch durch\-gäng\-igen und konsistenten Bildung von PSM Graphen im Zusammenspiel mit Szenarien und Kritikalitätsphänomenen aus Gefährdungsanalysen eine Rolle. In dem Kontext soll eine systematische Ableitung von Anforderungen und Spezifikationen für die sicherheitsgerichtete Systemarchitektur in Anbetracht des erforderlichen Systemverhaltens erreicht werden. Ebenso ist es eine Fragestellung, wie das PSM Konzept umgesetzt werden kann, um damit eine zielführende Definition der Gesamtheit der jeweils situativ angemessenen Fahrzeugverhaltensweisen darstellen zu können.

Letztlich werden eine systematische Vorgehensweise und ein Bewertungsverfahren benötigt, mit welchem festgestellt werden kann, inwieweit die bis dato konzipierten Fähigkeiten, sowie die technisch spezifizierten Funktionen ein ausreichend risikominimales Agieren des Fahrzeugs im Verkehrsgeschehen gewährleisten, bzw. welche Verbesserungen oder Absicherungen die Unzuläng\-lich\-keiten ausreichend mindern können. Dies ist letztlich die Grundlage für eine Überprüf\-bar\-keit des Fahrzeugverhaltens, sowie dessen Testbarkeit und Verifikation. 

Das Modell selbst wird ebenfalls kontinuierlich weiterentwickelt. Hier sind vor allem die Konkatenation von Regeln und Regelhierarchien zu nennen. Ersteres soll die Möglichkeit untersuchen, dass durch eine Regel die Ausdifferenzierung eines Sachverhaltes erfolgt, welcher dann Gegenstand einer weiteren Regel ist. Regelhierarchien könnten dazu dienen, Situationen abzufangen, wenn keine spezielle Regel einer Ebene anwendbar ist. In dieser Weise wird angestrebt, Verhalten bzw. Sollverhalten des Automaten immer präszier und vollständiger zu beschreiben.

\section{Danksagung}
Die Arbeiten an dieser Veröffentlichung wurden durch das Bundesministerium für Wirtschaft und Energie innerhalb des Forschungsprojekts \glqq Verifikations- und Validierungsmethoden automatisierter Fahrzeuge im urbanen Umfeld\grqq{} gefördert. Die Autoren bedanken sich bei den Projektpartnern für die erfolgreiche Zusammenarbeit, die konstruktiven Kritiken und Gespräche im Zuge der Erarbeitung des PSMs.

\pagebreak
\bibliography{literature}
\bibliographystyle{plain}

\end{document}